\documentclass[letterpaper, 10 pt, conference]{ieeeconf}  
\usepackage{amsmath} 
\usepackage{amssymb}  
\usepackage{graphicx}
\usepackage{amsmath}
\usepackage{epstopdf}
\usepackage{mathtools}
\usepackage{dsfont}
\usepackage{url}

\usepackage[ruled,norelsize]{algorithm2e}
\usepackage{color}
\usepackage{booktabs} 
\usepackage{multirow}   
\usepackage{balance}    
\usepackage{soul}   
\usepackage{hyperref}   
\usepackage[noadjust]{cite} 
\usepackage{bm}

\makeatletter
\newcommand{\removelatexerror}{\let\@latex@error\@gobble}
\makeatother

\newcommand{\beq}{\begin{equation}}
\newcommand{\eeq}{\end{equation}}
\newcommand{\rr}{{\mathbb R}}

\newcounter{algorithmctr}[section]
\renewcommand{\thealgorithmctr}{\thesection.\arabic{algorithmctr}}
   {\refstepcounter{algorithmctr}\begin{list}{}{%
       \setlength{\rightmargin}{0\linewidth}%
       \setlength{\leftmargin}{.05\linewidth}
        \setlength{\itemsep}{1pt}
  \setlength{\parskip}{0pt}
  \setlength{\parsep}{0pt}}%
       \rmfamily\small
       \item[]{\setlength{\parskip}{0ex}\hrulefill\par%
        \nopagebreak{\bfseries\textsf{Algorithm \thealgorithmctr~}}}}%
   {{\setlength{\parskip}{-1ex}\nopagebreak\par\hrulefill} \end{list}}
\IEEEoverridecommandlockouts
\newtheorem{assumption}{Assumption}

\newtheorem{prop}{Proposition}

\newtheorem{remark}{Remark}

\title{\LARGE \bf Simple Policy Evaluation for Data-Rich Iterative Tasks}

\author{Ugo Rosolia, Xiaojing Zhang, and Francesco Borrelli
\thanks{U.\ Rosolia, X.\ Zhang and F.\ Borrelli are with the Department of Mechanical Engineering, University of California at Berkeley ,
        Berkeley, CA 94701, USA
        {\tt\small\{ugo.rosolia, xiaojing.zhang,\ fborrelli\}{@}berkeley.edu}}%
}

\begin{document}

\maketitle
\thispagestyle{empty}
\pagestyle{empty}

\begin{abstract}
A data-based policy for iterative control task is presented. The proposed strategy is model-free and can be applied whenever safe input and state trajectories of a system performing an iterative task are available. These trajectories, together with a user-defined cost function, are exploited to construct a piecewise affine approximation to the value function. The approximated value function is then used to evaluate the control policy by solving a  linear program. We show that for linear system subject to convex cost and constraints, the proposed strategy guarantees closed-loop constraint satisfaction and performance bounds for the closed-loop trajectory. We evaluate the proposed strategy in simulations and experiments, the latter carried out on the Berkeley Autonomous Race Car (BARC) platform. We show that the proposed strategy is able to reduce the computation time by one order of magnitude while achieving the same performance as our model-based control algorithm.
\end{abstract}

\section{Introduction}








In the past decades researchers have focused on iterative strategies to synthesize control policy \cite{bristow2006survey, c33, c34, c35, c6, c4, liu2013nonlinear, LMPC, mania2018simple, schulman2015trust, wu2017scalable, mohajerin2018infinite}. The main idea is to execute the control task or a part of it repeatedly, and use the closed-loop data to automatically update the control policy. Each task execution is often referred to as ``trials" or ``iterations" and it may be performed in simulation or experiment. It is generally required that at each update the control policy guarantees safety. Furthermore, it is desirable that the closed-loop performance improves at each policy update and that the iterative scheme converges to a (local) optimal steady state behavior. Algorithms that iteratively update the control policy and satisfy the above properties have been extensively studied in the literature.

Iterative Learning Control (ILC) is a control strategy that allows learning from previous iterations to improve the closed-loop tracking performance \cite{bristow2006survey}. In ILC, at each iteration, the system starts from the same initial condition and the controller objective is to track a given reference, rejecting periodic disturbances. The main advantage of ILC is that information from previous iterations are incorporated into the problem formulation at the next iteration, in order to improve the control policy while guaranteeing safety. Furthermore, it is possible to show that as the number of iteration increases the control policy converges to a steady state (local) optimal behavior \cite{c33, c34, c35, c6, c4, liu2013nonlinear}. Recently, we proposed an ILC algorithm called Learning Model Predictive Control (LMPC), where the controller's goal is to minimize a generic positive definite cost function \cite{LMPC}. At each time, the LMPC solves a finite time optimal control problem, where the data from previous iterations are used to update the terminal constraint and terminal cost, which approximates the value function. In the above mentioned ILC schemes, the data from each iteration are used to update the control policy while guaranteeing safety and performance improvement. However, the computational complexity of these algorithms does not decrease when the policy update has converged, although the controller applies the same or similar control actions at each iteration. Indeed, evaluating the control policy involves the solution to a model-based optimization problem. In this work we propose a model-free data-based policy, which may be used to reduce the computational burden of ILC algorithms which have reached convergence.

Model-free iterative algorithms, such as policy search and $Q$-learning, have recently gained popularity. In policy search, the control policy is updated using derivative-free optimization \cite{recht2018tour} or gradient estimation \cite{schulman2015trust}. These algorithms have been successfully tested in simulation scenarios to perform complex locomotions tasks. For more details we refer to \cite{recht2018tour, mania2018simple, schulman2015trust, wu2017scalable, mohajerin2018infinite}. $Q$-learning is an approximate dynamic programming strategy where an optimal cost function for a state input pair is learned from data \cite{bertsekas2005dynamic, recht2018tour}. The optimal cost function is usually approximated using a linear mapping of a state dependent feature vector. These features may be arbitrary nonlinear functions of the states, see \cite[Chapter VI]{bertsekas2005dynamic} for details. In $Q$-learning, the policy is evaluated minimizing the approximated value function at the current state with respect to the control input \cite[Chapter VI]{bertsekas2005dynamic}, \cite{recht2018tour}.

In all the aforementioned literature, it is important to distinguish  between the strategy used to update the control policy and the method used to evaluate the current policy. This paper focuses on the latter problem.  We propose a simple, perhaps the simplest, value function approximation strategy, which may be used to compute a control law from historical state-input data, regardless  on the techniques used to generate the data. We build on~\cite{linearLMPC} where we exploit stored input and state trajectories along with a user-defined cost to construct a piecewise-affine approximation of the value function. The value function approximation is defined as a convex combination of the cost associated with the stored closed-loop trajectories. In the present work, we propose to exploit the multipliers from the convex combination of the cost to extract the control action from the stored inputs. The proposed strategy needs to store the input and state trajectories and may not be applied when limited memory storage is available. Furthermore, we proposes a local approximation of the value function, which allows to further reduce the computational burden of the proposed policy evaluation method. Finally, we show that for linear systems subject to convex cost and convex constraints, the data-based policy guarantees safety, stability and performance bounds. We evaluate the proposed strategy on the Berkeley Autonomous Race Car (BARC) platform, and demonstrate that the data-based policy is able to match the performance of our model-based ILC algorithm, while being almost $30$x faster at computing the control inputs.

The paper is organized as follows: in Section II we introduce the problem formulation. In Section III we describe the proposed approach. First we show how to use data to construct the safe set and the value function approximation. Afterwards, we introduce the control design. The properties of the proposed approach are discussed in Section IV. Finally, in Section V we test the proposed data-based policy on simulation and experiment, the latter on the Berkeley Autonomous Race Car (BARC) platform.

\section{Problem Formulation}
Consider the unknown deterministic system
\begin{equation} \label{eq:System}
    x_{t+1} = A x_t + B u_t
\end{equation}
where $x_t \in \mathbb{R}^n$ and $u_t \in \mathbb{R}^d$ are the system's state and input, respectively. Furthermore, the system is subject to the following state and input constraints,
\begin{equation} \label{eq:stateInputConstr}
    x_t \in \mathcal{X} \text{ and } u_t \in \mathcal{U}, ~\forall t \in \{0,\ldots, T \}
\end{equation}
where $T$ is the time as which the control task is completed.

In the following we assume that closed-loop state and input trajectories starting at different initial states $x_0$ are stored. In particular, 
for $j \in \{0, \ldots, M\}$ we are given the following input sequences 
\begin{equation}\label{eq:givenIinputs}
\begin{aligned}
    {\bf{u}}^j &= [u_0^j , \ldots, u_{T_j}^j]
\end{aligned}
\end{equation}
and the associated closed-loop trajectories
\begin{equation}\label{eq:givenClosedLoop}
\begin{aligned}
    {\bf{x}}^j &= [x_0^j , \ldots, x_{T_j}^j] \\
\end{aligned}
\end{equation}
where $x_{t+1}^j = A x_t^j + B u_t^j$ and $T_j$ is the time at which the task is completed. 
These trajectories will be used to design a data-based policy for the unknown system \eqref{eq:System}.

Finally, we defined the cost-to-go associated with the $j$th closed-loop trajectory
\begin{equation} \label{eq:RelalizedCost}
    J^j\big(x_0^j\big)= \sum_{k = 0}^{T_j} h(x^j_k, u^j_k),
\end{equation}
where $x_k^j$ and $u_k^j$ are the stored state and applied input to system \eqref{eq:System} at time $k$ of the $j$th iteration.



\begin{assumption}\label{ass:feasibility}
All $M+1$ input and state sequences in \eqref{eq:givenIinputs}-\eqref{eq:givenClosedLoop} are feasible and known. Furthermore, assume that the state sequence in~\eqref{eq:givenClosedLoop} converges to the origin and the terminal input $u_{T_j}^j=0$.
\end{assumption}

\begin{remark}
We have decided to focus on the linear systems \eqref{eq:System} as this will allow us to rigorously characterize the properties of the proposed approach. However, we underline that the computational cost associated with the proposed strategy is independent on the linearity of the controlled system. Thus, the proposed strategy can be implemented also on nonlinear systems as shown in Section~\ref{sec:expResults}.
\end{remark}

\begin{remark}
We have decided to consider a regulation problem to streamline the presentation of the paper. In the Appendix, we show that the proposed strategy can be used to steer system~\eqref{eq:System} to a terminal control invariant set $\mathcal{X}_F$,  without losing guarantees on safety and performance.
\end{remark}

\section{Proposed Approach}
In this section we describe the proposed approach. First, we introduce the sampled safe set and value function approximation computed from data, which were first introduced in  \cite{LMPC} and \cite{linearLMPC}. Afterward, we show how these quantities are used to evaluate the data-based policy.

\subsection{Safe Set}
We define the collection of the $M$ closed-loop trajectories in \eqref{eq:givenClosedLoop} as the sampled \textit{Safe Set}, 
\begin{equation} 
    \mathcal{SS} = \bigcup_{j = 0}^M \bigcup_{t = 0}^{T_j} x_t^j. \notag
\end{equation}
Notice that for all $x \in \mathcal{SS}$, it exists a sequence of control actions that can steer the system to the origin \cite{LMPC}. Finally, we define the \textit{convex safe set} $\mathcal{CS}$ as
\begin{equation}\label{eq:CS}
    \mathcal{CS} = \text{Conv}\big( \mathcal{SS} \big).
\end{equation}
$\mathcal{CS}$ will be used in the next section to defined the domain of the approximation to the value function.

\subsection{Q-function}
In this section we show how the stored data in~\eqref{eq:givenIinputs} and \eqref{eq:givenClosedLoop} are used to approximate the value function. First, given the stored states ${\bf{x}}^j$ and inputs ${\bf{u}}^j$ for $j \in \{0, \ldots, M\}$, we define the cost-to-go associated with each stored state $x_k^j$,
\begin{equation}
    J^j_k(x_k^j) = \sum_{i=k}^{T_j} h(x_i^j, u_i^j). \notag
\end{equation}
The realized cost-to-go $J^j_k(x_k^j)$ is used to compute the \textit{Q-function} defined as
\begin{equation}\label{eq:valueFunc}
    \begin{aligned}
        Q(x) = \min_{ \bm\lambda \geq 0} \quad & \sum_{j=0}^{M}\sum_{k=0}^{T_j} \lambda_k^{j} J^j(x_k^j) \\
        \text{s.t.}\quad  &\sum_{j=0}^{M}\sum_{k=0}^{T_j} \lambda_k^{j} = 1,\\
        &\sum_{j=0}^{M}\sum_{k=0}^{T_j} \lambda_k^{j} x_k^j = x.
    \end{aligned}
\end{equation}
where $\bm\lambda = [\lambda_{0}^0, \ldots,\lambda_{T_0}^0,  \ldots, \lambda_{0}^M,, \ldots,\lambda_{T_M}^M]$. The Q-function $Q(\cdot)$ interpolates the realized cost-to-go over the convex safe set. Moreover, we underline that Problem~\eqref{eq:valueFunc} is a parametric LP and therefore $Q(x)$ is a  piecewise affine function of $x$ \cite{borrelli2017predictive}. Finally, we notice that the domain of $Q(\cdot)$ is the convex safe set $\mathcal{CS}$, indeed $\forall x \notin \mathcal{CS}$ the optimization problem \eqref{eq:valueFunc} is not feasible.

\subsection{Data-Based Policy}
We are finally ready to introduce the data-based policy. At each time $t$, we evaluate the approximation to the value function \eqref{eq:valueFunc} at the current state $x_t$, solving the following optimization problem,
\begin{equation}\label{eq:valueFuncEval}
    \begin{aligned}
        Q(x_t) = \min_{\bm\lambda_t \geq 0} \quad & \sum_{j=0}^{M}\sum_{k=0}^{T_j} \lambda_{k|t}^{j} J_k^j(x_k^j) \\
        \text{s.t.}\quad &\sum_{j=0}^{M}\sum_{k=0}^{T_j} \lambda_{k|t}^{j} = 1,\\
        &\sum_{j=0}^{M}\sum_{k=0}^{T_j} \lambda_{k|t}^{j} x_k^j = x_t. 
    \end{aligned}
\end{equation}
where $\bm\lambda_t = [\lambda_{0|t}^0, \ldots,\lambda_{T_0|t}^0,\ldots, \lambda_{0|t}^M, \ldots
\lambda_{T_M|t}^M]$.\\ Let 
\begin{equation}\label{eq:optimalSol}
    \bm\lambda_t^* = [\lambda_{0|t}^{0,*}, \ldots, \lambda_{k|t}^{j,*}, \ldots, \lambda_{T_M|t}^{M,*}]
\end{equation}
be the optimal solution at time $t$ to \eqref{eq:valueFuncEval}, then we apply to system \eqref{eq:System} the following input
\begin{equation}\label{eq:policy}
    u_t = \pi(x_t) = \sum_{j=0}^{M}\sum_{k=0}^{T_j} \lambda_{k|t}^{j,*} u_k^j.
\end{equation}
Basically, the data-based policy \eqref{eq:valueFuncEval} and \eqref{eq:policy} computes the control input $u_t$ as the weighted sum of stored inputs, where the weights are the solution to the minimization problem \eqref{eq:valueFuncEval}. 

\subsection{Local Data-Based Policy}
In this section we propose a Local Data-Based policy which can be used to limit the computational burden of problem \eqref{eq:valueFuncEval}, when a considerable amount of data is given. First, we define the local $Q$-function $Q_L(\cdot)$ as
\begin{equation}\label{eq:localValueFunc}
    \begin{aligned}
        Q_L(x_t) = \min_{ \bm \lambda_t \geq 0 } \quad & \sum_{j=0}^{M}\sum_{k \in \mathcal{K}^j(x)} \lambda_{k|t}^{j} J_k^j(x_k^j) \\
        \text{s.t.}\quad  & \sum_{j=0}^{M}\sum_{k \in \mathcal{K}^j(x)} \lambda_{k|t}^{j} = 1,\\
        &\sum_{j=0}^{M}\sum_{k \in \mathcal{K}^j(x)} \lambda_{k|t}^{j} x_k^j = x_t
        \end{aligned}
\end{equation}
where $\bm \lambda_t = [\lambda_{t^{0,*}_1|t}^0, \ldots,\lambda_{t^{0,*}_N|t}^0, \ldots, \lambda_{t^{M,*}_1|t}^0, \ldots, \lambda_{t^{M,*}_N|t}]$. The elements of the set $\mathcal{K}^j(x) = \{t^{j,*}_1, \ldots, t^{j,*}_N\}$ are defined as
\begin{equation*}
\begin{aligned}
    [t^{j,*}_1, \ldots, t^{j,*}_N] = \arg \min_{\mathbf t}   & \quad \sum_{l = 1}^N ||x^j_{t_l} - x||_2 \\
    \text{s.t.}  & \quad t_i \neq t_j,~ \forall i \neq j \\
    & \quad t_i \in \{0, \ldots, T_j \}, \forall i \in \{0, \ldots,N \}.
\end{aligned}
\end{equation*}
For the $j$-th trajectory, the set $\mathcal{K}^j(x)$ collects the indices of the $N$ closest point to the state $x$. Notice that $N \leq \max_{i\in\{0,\ldots,j\}} T_i$ is a user-defined parameter.

Finally, we define the local data-based policy where at each time $t$ we solve $Q_L(x_t)$ in~\eqref{eq:localValueFunc}. Then, given the optimal solution $ \bm \lambda_t^*$ to Problem~\eqref{eq:localValueFunc}, we apply the following 
input 
\begin{equation}\label{eq:localPolicy}
        u_t = \pi(x_t) = \sum_{j=0}^{M}\sum_{k \in \mathcal{K}^j(x_t)} \lambda_{k|t}^{j,*} u_k^j
\end{equation}
to system \eqref{eq:System}.

\section{Properties}
In this section we analyze the properties of the proposed data-based policy  \eqref{eq:valueFuncEval} and \eqref{eq:policy}. We show that the proposed strategy guarantees safety, closed-loop stability and  performance bounds.

\begin{prop}
\textit{(Feasibility)} Consider the closed-loop system~\eqref{eq:System} and \eqref{eq:policy}. Let Assumptions~\ref{ass:feasibility} hold and $\mathcal{CS}$ be the convex safe set defined in \eqref{eq:CS}. If the initial state $x_0 \in \mathcal{CS}$. Then, the data-based policy \eqref{eq:valueFuncEval} and \eqref{eq:policy} is feasible for all time $t\geq0$.
\end{prop}
\begin{proof}
The proof follows from linearity of the system. \\
We assume that at time $t$ the system state $x_t \in \mathcal{CS}$, therefore the optimization problem \eqref{eq:valueFuncEval} is feasible. Let \eqref{eq:optimalSol} be the optimal solution to \eqref{eq:valueFuncEval}, then at the next time step $t+1$ we have 
\begin{equation*}
\begin{aligned}
         x_{t+1} &=  A x_t + B \sum_{j=0}^M\sum_{k=0}^{T_j} \lambda^{j,*}_{k|t} u^j_k  \\
    &  = A \sum_{j=0}^M\sum_{k=0}^{T_j} \lambda^{j,*}_{k|t} x^j_k + B \sum_{j=0}^M\sum_{k=0}^{T_j} \lambda^{j,*}_{k|t} u^j_k \\
    &  = \sum_{j=0}^M\sum_{k=0}^{T_j} \lambda^{j,*}_{k|t} (A  x^j_k + B u_k^j) \in \mathcal{CS}.
\end{aligned}
\end{equation*}
By Assumption~\ref{ass:feasibility} we have that 
\begin{equation*}
    \sum_{j=0}^M \lambda^{j,*}_{T_j|t} (A  x^j_{T_j} + B u_{T_j}^j)=0
\end{equation*} 
and therefore
\begin{equation*}
    x_{t+1} =  \sum_{j=0}^M\sum_{k=0}^{T_j} \lambda^{j,*}_{k|t} (A  x^j_k + B u_k^j) = \sum_{j=0}^M\sum_{k=0}^{T_j} \bar \lambda^{j}_k x^j_k
\end{equation*}
where $\forall j\in \{0, \ldots M\}$ 
\begin{equation}
    \begin{aligned}\label{eq:feasibleSol}
        &\bar \lambda_0^j = 0, \\
        &\bar \lambda_{k_j}^j = \lambda_{k_j-1|t}^{j,*}, \quad \quad \quad \quad \quad \forall k_j \in \{ 1, \ldots, T_j-1 \} \\
        &\bar \lambda_{T_j}^j = \lambda_{T_j-1|t}^{j,*} + \lambda_{T_j|t}^{j,*}
    \end{aligned}
\end{equation}
is a feasible solution to the optimization problem \eqref{eq:valueFuncEval} at time $t+1$.\\
By assumption we have that at time $t=0$ the state $x_0 \in \mathcal{CS}$. Furthermore, we have shown that if at time $t$ the state $x_t \in \mathcal{CS}$, then at time $t+1$ the state $x_{t+1} \in \mathcal{CS}$ and the optimization problem \eqref{eq:valueFuncEval} is feasible. Therefore by induction we conclude that $x_t \in \mathcal{CS} \subseteq \mathcal{X}, ~\forall t \in \mathbb{Z}_{0+}$ and that the optimization problem \eqref{eq:valueFuncEval} is feasible $\forall t \in \mathbb{Z}_{0+}$.
\end{proof}

The above \textit{Proposition 1} implies that the data-based policy \eqref{eq:valueFuncEval} and \eqref{eq:policy} satisfies the input constraints, and the closed-loop system \eqref{eq:System} and \eqref{eq:policy} satisfies the state constraints at all time instants, i.e. $u_t \in \mathcal{U}$ and $x_t \in \mathcal{X}, ~\forall t \in \mathbb{Z}_{0+}$.

\begin{assumption}\label{ass:cost}
The stage cost $h(\cdot, \cdot)$ is a continuous convex function and $\forall u \in \mathcal{U}$ it satisfies
\begin{equation}
\begin{aligned}
h(0,u) = 0,\textrm{ and}~ h(x,u) \succ 0 ~ \forall ~ x \in&~\rr^n \setminus \{0\}. \notag
\end{aligned}
\end{equation}
\end{assumption}

\begin{prop}
\textit{(Convergence)} Consider the closed-loop system~\eqref{eq:System} and \eqref{eq:policy}. Let Assumptions~\ref{ass:feasibility}-\ref{ass:cost} hold and $\mathcal{CS}$ be the convex safe set defined in \eqref{eq:CS}. If the initial state $x_0 \in \mathcal{CS}$. Then, the origin of the closed-loop system \eqref{eq:System} and \eqref{eq:policy} is asymptotically stable.
\end{prop}
\begin{proof}In the following we show that the approximated value function $Q(\cdot)$ from~\eqref{eq:valueFuncEval} is a Lyapunov function for the origin of the closed loop system \eqref{eq:System} and \eqref{eq:policy}. Continuity of $Q(\cdot)$ can be shown as in \cite[Chapter 7]{borrelli2017predictive}. Moreover from \eqref{eq:RelalizedCost} and Assumption~2 we have that $Q(x) \succ 0 ~ \forall ~ x \in \mathcal{CS} \setminus \{0\}$ and $Q(0)=0$. Thus, we need to show that  $Q(\cdot)$ is decreasing along the closed loop trajectory.\\
By feasibility of Problem~\eqref{eq:valueFuncEval} from Theorem~1, we have that at time $t$
\begin{equation}\label{eq:lyapunovPart1}
    \begin{aligned}
        Q(x_t) &= \sum_{j=0}^M \sum_{k=0}^{T_j} \lambda_{k|t}^{j,*} {J}_k^j(x_k^j) = \sum_{j=0}^M \sum_{k=0}^{T_j} \lambda_{k|t}^{j,*} \sum_{i=k}^{T_j} h(x_i^j, u_i^j) \\
        & = \sum_{j=0}^M \sum_{k=0}^{T_j} \lambda_{k|t}^{j,*} h(x_k^j, u_k^j) + \sum_{j=0}^M \sum_{k=0}^{T_j-1} \lambda_{k|t}^{j,*} {J}_{k+1}^j(x_{k+1}^j).
    \end{aligned}
\end{equation}
We notice that the summation of the cost-to-go in the above expression can be rewritten as
\begin{equation}\label{eq:lyapunovPart2}
    \sum_{j=0}^M \sum_{k=0}^{T_j-1} \lambda_{k|t}^{j,*} {J}_{k+1}^j(x_{k+1}^j) = \sum_{j=0}^M \sum_{k=0}^{T_j} \bar \lambda_{k|t}^{j} {J}_k^j(x_k^j) \geq Q(x_{t+1}),
\end{equation}
where $\bar \lambda_{k|t}^{j}$ is the candidate solution defined in \eqref{eq:feasibleSol}.

Finally, from equations \eqref{eq:lyapunovPart1} and \eqref{eq:lyapunovPart2} we conclude that the optimal cost is a decreasing Lyapunov function along the closed loop trajectory,
\begin{equation}\label{eq:LyapProof2}
	\begin{aligned}
		Q(x_{t+1})-Q(x_{t}) \leq - & \sum_{j=0}^M \sum_{k=0}^{T_j} \lambda_{k|t}^{j,*} h(x_k^j, u_k^j) < 0, \\
		&~~~~~~~~~~~~\forall~x_t \in R^n \setminus \{0\}.
	\end{aligned}
\end{equation}
Equation (\ref{eq:LyapProof2}), the positive definitiveness of $h(\cdot, \cdot)$ and the continuity of $Q(\cdot)$ imply that the origin of the closed-loop system~\eqref{eq:System} and \eqref{eq:policy} is asymptotically stable. 
\end{proof}

\begin{prop}
\textit{(Cost)} Consider the closed-loop system~\eqref{eq:System} and \eqref{eq:policy}. Let Assumptions~\ref{ass:feasibility}-\ref{ass:cost} hold and $\mathcal{CS}$ be the convex safe set defined in \eqref{eq:CS}. If the initial state $x_0 \in \mathcal{CS}$. Then, the $Q$-function at $x_0$, $Q(x_0)$, upper bounds the cost associated with the trajectory of closed-loop system~\eqref{eq:System} and \eqref{eq:policy},
\begin{equation}\label{eq:closedLoopCost}
    J\big(x_0\big) =   \sum_{k=0}^{\infty} h(x_k, u_k) \leq Q\big(x_0\big)
\end{equation}
where $\{x_0 , \ldots, x_{t}, \ldots\}$ and $ \{u_0 , \ldots, u_{t}, \ldots\}$ are the closed-loop trajectory and associated input sequence, respectively. 

\end{prop}
\begin{proof} From \eqref{eq:LyapProof2} and convexity of $h(\cdot, \cdot)$, we have that 
\begin{equation}
    Q(x_t) \geq h(x_t, u_t) + Q(x_{t+1}) \notag
\end{equation}
Using the above equation recursively and from the asymptotic convergence to the origin we have that
\begin{equation}
\begin{aligned}
    Q(x_0) & \geq h(x_0, u_0) + Q(x_{1})   \\
   & \geq \sum_{k=0}^{\infty} h(x_k, u_k) + \lim_{k \rightarrow \infty} Q(x_{k})= \sum_{k=0}^{\infty} h(x_k, u_k). \notag
\end{aligned}
\end{equation}
\end{proof}

Note that, if the optimal closed-loop trajectory from $x_0=x_s$ is given, then the approximated value function $Q(x_s)$ will be the optimal cost-to-go from $x_s$. Consequently, \textit{Proposition~3} implies that the proposed data-based policy will behave optimally for $x_0=x_s$, if the optimal behavior from $x_0=x_s$ has been observed. 


\section{Examples}
In this section we first test the data-based policy \eqref{eq:valueFuncEval} and \eqref{eq:policy} on a double integrator system. Afterwards, we test the local data-based policy \eqref{eq:localValueFunc} and \eqref{eq:localPolicy} on the Berkeley Autonomous Racing Car (BARC) platform.  

\subsection{Example I: Double Integrator}

Consider the following discrete time Constrained Linear Quadratic Regulator (CLQR) problem
\begin{equation}\label{eq:CLQR}
\begin{aligned}
J^*\big(x_0\big) =\min_{\bar u_0, \bar  u_1,\ldots} & \quad \sum\limits_{k=0}^{\infty} \Big[ ||\bar x_k||_2^2 + ||\bar u_k||_2^2 \Big] \\
\textrm{s.t.}&\\
   &\quad \bar x_{k+1}= \begin{bmatrix} 1 & 1 \\ 0 & 1 \end{bmatrix} \bar x_k +  \begin{bmatrix} 0 \\ 1 \end{bmatrix} \bar u_k,~\forall k\geq 0 \\
   &\quad  \begin{bmatrix} -10 \\ -10 \end{bmatrix} \leq \bar x_k \leq \begin{bmatrix} 10 \\ 10 \end{bmatrix} ~ \forall k\geq 0 \\
   &\quad -1 \leq \bar u_k \leq 1 ~~\forall k\geq 0, \\
  &\quad \bar x_0=x_0=[-1, 3]^\top.
\end{aligned}
\end{equation}
First, we construct the convex safe set using one solution to the above CLQR and we empirically validate \textit{Proposition}~1-3. Afterwards, we analyze the effect of the amount of data on the value function approximation and the data-based policy \eqref{eq:valueFuncEval} and \eqref{eq:policy}.

\subsubsection{Properties verification}
First, we compute and store the optimal solution to the CLQR problem \eqref{eq:CLQR},
\begin{equation}\label{eq:closedLoopLQR}
    \begin{aligned}
             [\bar x_0^*, \bar x_1^*, \ldots, \bar x_T^*] \\
             [\bar u_0^*, \bar u_1^*, \ldots, \bar u_T^*]
    \end{aligned}
\end{equation}
where $T$ is the time index at which $||\bar x_T^*||_2^2 \leq \epsilon = 10^{-10}$.

\begin{figure}[h!]
    \centering
	\includegraphics[width= \columnwidth]{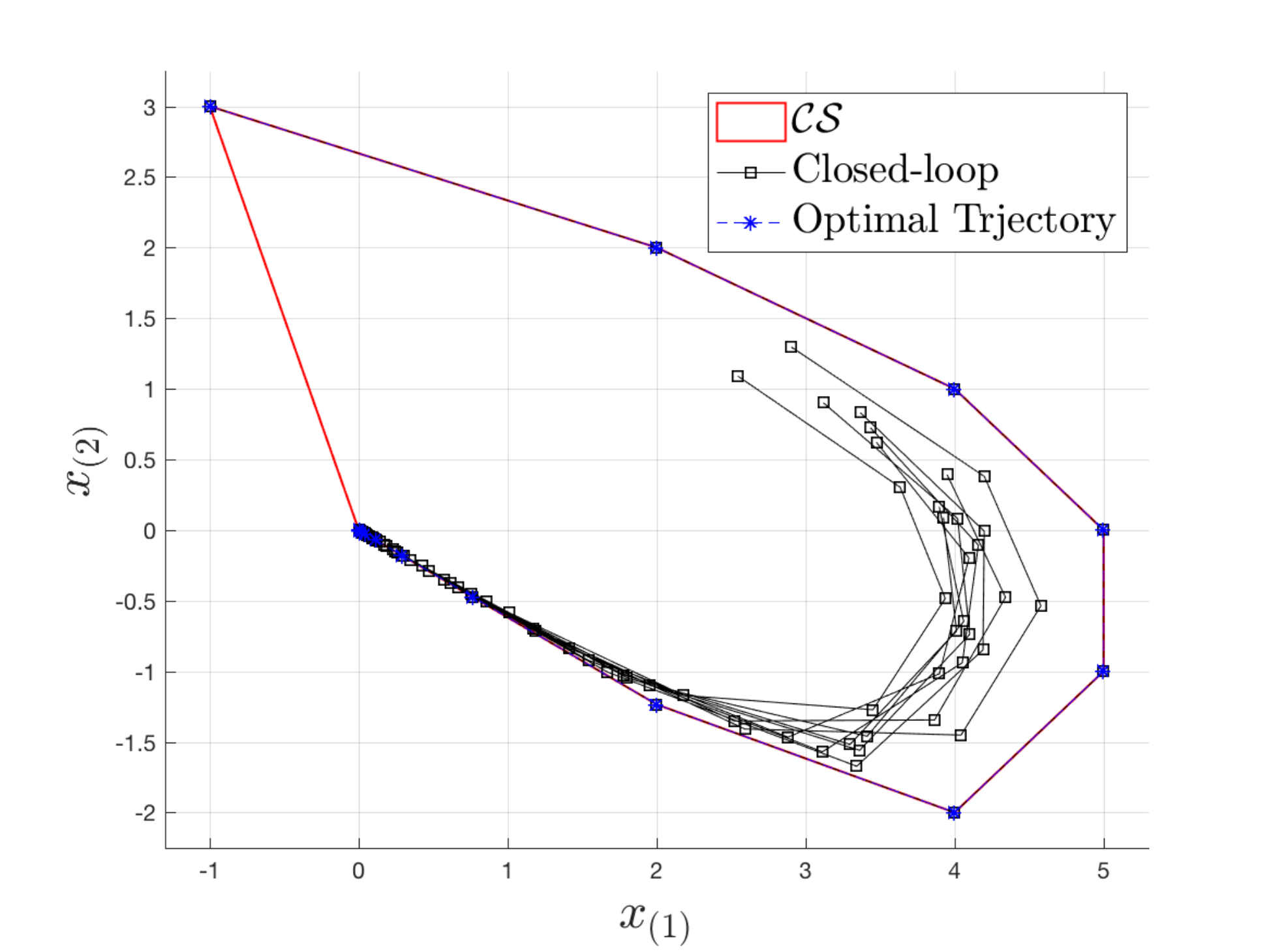}.
    \caption{Closed-loop trajectories performed by the data-based policy.}
    \label{fig:closedLoopLQR}
\end{figure}

The stored optimal trajectory in \eqref{eq:closedLoopLQR} is used to build the convex safe set $\mathcal{CS}$ in \eqref{eq:CS} and the approximation to the value function $Q(\cdot)$ in \eqref{eq:valueFunc}. We tested the data-based policy for $x_0 = \bar x_0^*$ and for other $10$ randomly picked initial conditions inside $\mathcal{CS}$. We denote the resulting closed-loop trajectories and associated input sequences for $j \in \{0, \ldots, 9 \}$ as
\begin{equation}\label{eq:givenStateAndInputResults}
\begin{aligned}
    {\bf{x}}^j &= [x_0^j , \ldots, x_{T_j}^j]\\
    {\bf{u}}^j &= [u_0^j , \ldots, u_{T_j}^j] \\
\end{aligned}.
\end{equation}
Figure~\ref{fig:closedLoopLQR} shows the closed-loop trajectories, we confirm that state and input constraints are satisfies, accordingly to \textit{Proposition~1}. Furthermore, we notice that the closed-loop trajectories converge to the origin as we expected from \textit{Proposition~2}. It is interesting to notice that for $x_0 = \bar x_0^*$ the closed-loop trajectory performed by the data-based policy overlaps with the optimal one.

Moreover, we analyze the cost associated with the closed-loop trajectories~\eqref{eq:closedLoopCost}. Table~\ref{table:comparisonLQR} shows the realized cost~\eqref{eq:closedLoopCost} and the approximated value function $Q(\cdot)$ evaluated at different initial conditions. We confirm that $Q\big(x_0\big)$ upper bounds the performance of the closed-loop trajectory, as shown in \textit{Proposition 3}.
   
\begin{table}[h!]
\centering\caption{Comparison of the realized cost and value function for different initial conditions}\label{table:comparisonLQR}
\begin{tabular}{l|l|l}\toprule
 $~~~~~~~~~~x_0$ & $J\big(x_0\big)$  & $Q\big(x_0\big)$ \\ \midrule
 $[-1, 3]^\top$&  $112.53$ & $112.53$   \\
 $[2.9033,   1.2959]^\top$ & $78.60$ & $89.60$   \\
 $[3.9495,   0.3921]^\top$&  $62.00$ & $73.97$   \\
 $[3.3673,   0.8315]^\top$&  $66.45$ & $79.23$   \\
 $[3.4349,   0.7243]^\top$&  $62.96$ & $76.79$   \\
 $[3.9253,   0.0874]^\top$&  $50.37$ & $63.69$   \\
 $[3.1189,   0.9013]^\top$&  $63.11$ & $78.18$   \\
 $[3.8963,   0.1645]^\top$&  $52.12$ & $65.74$  \\
 $[2.5449,   1.0898]^\top$&  $58.04$ & $76.85$   \\
 $[3.4751,   0.6212]^\top$&  $59.22$ & $74.06$   \\
 $[2.5770,   1.1763]^\top$&  $63.34$ & $80.50$  \\ \bottomrule
\end{tabular}
\end{table}

\subsubsection{The effect of data} Finally, we empirically analyze the effect of data on the $Q$-function and the data-based policy. First, we construct two approximations to the value function: $Q^{1}(\cdot)$ using \eqref{eq:closedLoopLQR} and the $10$ stored state and input trajectories computed in the previous subsection \eqref{eq:givenStateAndInputResults}, and $Q^{2}(\cdot)$ using \eqref{eq:closedLoopLQR} and the optimal solution to the CLQR for $\bar x_0 = [2.9033,   1.2959]$. Afterwards, we run the data-based policy using $Q^1(\cdot)$ and $Q^2(\cdot)$. Table~II shows the cost associated with the closed-loop trajectories $J^i(\cdot)$ and the value function approximation $Q^i(\cdot)$, for $i = \{1,2\}$. We notice that $Q^1(x_0)$ lower bounds $Q(x_0)$ from Table~I and, therefore, better approximates the value function. However, the realized cost $J^1(x_0)$ does not improve with respect to $J(x_0)$ from Table~I. On the other hand, we notice that the data-based policy constructed using $Q^2(\cdot)$ is able to improve the closed-loop performance $J^2(x_0)$. It is interesting to notice that $Q^1(x_0)$ is constructed using one optimal trajectory and $10$ feasible trajectories, whereas $Q^2(x_0)$ is constructed using just two optimal trajectories. This result is interesting and it suggests that not all data points are equally valuable. 
   
\begin{table}[h!]
\centering\caption{Comparison of the realized cost and value function for different initial conditions}\label{table:comparisonLQR}
\begin{tabular}{l | ll | ll}
 \toprule
 $~~~~~~~~~~x_0$ & $J^{1}\big(x_0\big)$  & $Q^{1}\big(x_0\big)$ & $J^{2}\big(x_0\big)$  & $Q^{2}\big(x_0\big)$ \\ \midrule
 $[-1, 3]^\top$           &  $112.53$ & $112.53$ & $112.53$ & $112.53$  \\
 $[2.9033,   1.2959]^\top$&  $78.60$  & $78.60$  & $72.89$  & $72.89$  \\
 $[3.9495,   0.3921]^\top$&  $62.00$  & $62.00$  & $59.43$  & $62.12$  \\
 $[3.3673,   0.8315]^\top$&  $66.45$  & $66.45$  & $61.86$  & $66.39$  \\
 $[3.4349,   0.7243]^\top$&  $62.96$  & $62.96$  & $58.97$  & $64.38$  \\
 $[3.9253,   0.0874]^\top$&  $50.37$  & $50.37$  & $49.24$  & $54.57$  \\
 $[3.1189,   0.9013]^\top$&  $63.11$  & $63.11$  & $58.76$  & $65.04$  \\
 $[3.8963,   0.1645]^\top$&  $52.12$  & $52.12$  & $50.73$  & $55.86$ \\
 $[2.5449,   1.0898]^\top$&  $58.04$  & $58.04$  & $53.85$  & $62.65$  \\
 $[3.4751,   0.6212]^\top$&  $59.22$  & $59.22$  & $55.81$  & $62.12$  \\
 $[2.5770,   1.1763]^\top$&  $63.34$  & $63.34$  & $58.63$  & $65.74$ \\ \bottomrule
\end{tabular}
\end{table}

\subsection{Example II: Autonomous Racing}\label{sec:expResults}
In this Section we test the proposed control strategy on a 1/10-scale open source vehicle platform called the  Berkeley Autonomous Race Car (BARC)\footnote{More information at the project site \href{http://www.barc-project.com/}{barc-project.com}}. The BARC is equipped with an inertial measurement unit, encoders, and an ultrasound-based indoor GPS system. The vehicle has an Odroid XU4 which is used for collecting data and running the state estimator. Finally, the computation are performed on a MSI laptop with an intel CORE i7. A video of the experiments can be found here: {\footnotesize{ \url{https://youtu.be/pB2pTedXLpI}}}.

The control task is to drive the vehicle continuously around the track minimizing the lap time, while being within the track boundaries. The state vector is 
\begin{equation}
    x = [v_{x}, v_y, w_z, e_{\psi}, s, e_y]^\top \notag
\end{equation}
where $v_{x}, v_y$ and $w_z$ represent the vehicle's longitudinal, lateral and angular velocity in the body fixed frame. The position of the system is measured with respect to the curvilinear reference frame \cite{micaelli}, where $s$ represents the progress of the vehicle along the centerline of the track, $e_{\psi}$ and  $e_y$ represent the  heading  angle  and  lateral  distance error  between  the  vehicle  and  the  path. It is important to underline that, given the lane boundaries $e_{y_{min}}$ and $e_{y_{max}}$, the feasible region 
$\mathcal{X} = \{ x \in \mathbb{R}^n : e_{y_{min}} \leq e_6^\top x \leq e_{y_{max}}  \}$ for $e_6=[0,0,0,0,0,1]^\top$ is a convex set. The control input vector is $u=[\delta, a]$ where $\delta$ and $a$ are the steering angle and acceleration, respectively. The input constraints are 
\begin{equation}
\begin{aligned}
    -0.25 [\text{rad}] \leq &\delta \leq 0.25 [\text{rad}]\\
    -0.7 [\text{m/s}^2] \leq&  a \leq 2 [\text{m/s}^2]. \notag
\end{aligned}
\end{equation}
Finally, we underline that the autonomous racing problem is a repetitive task and the goal is not to steer the system to the origin. Therefore, we use the method from~\cite{cdcLMPC} to apply the proposed strategy to the autonomous racing repetitive control problem. In particular, we define the set of state beyond the finish line of the track of length $L$, $\mathcal{X}_F = \{ x \in \mathbb{R}^6 : e_5^\top x \geq L \}$ and we use the set $\mathcal{X}_F$ to compute the cost associated with the stored trajectories
\begin{equation*}
    h(x,u) = \begin{cases} 1 & \mbox{If } x \notin \mathcal{X}_F \\
    0 & \mbox{If } x\in \mathcal{X}_F \end{cases}.
\end{equation*}

\begin{figure}[h!]
    \centering
	\includegraphics[width= \columnwidth]{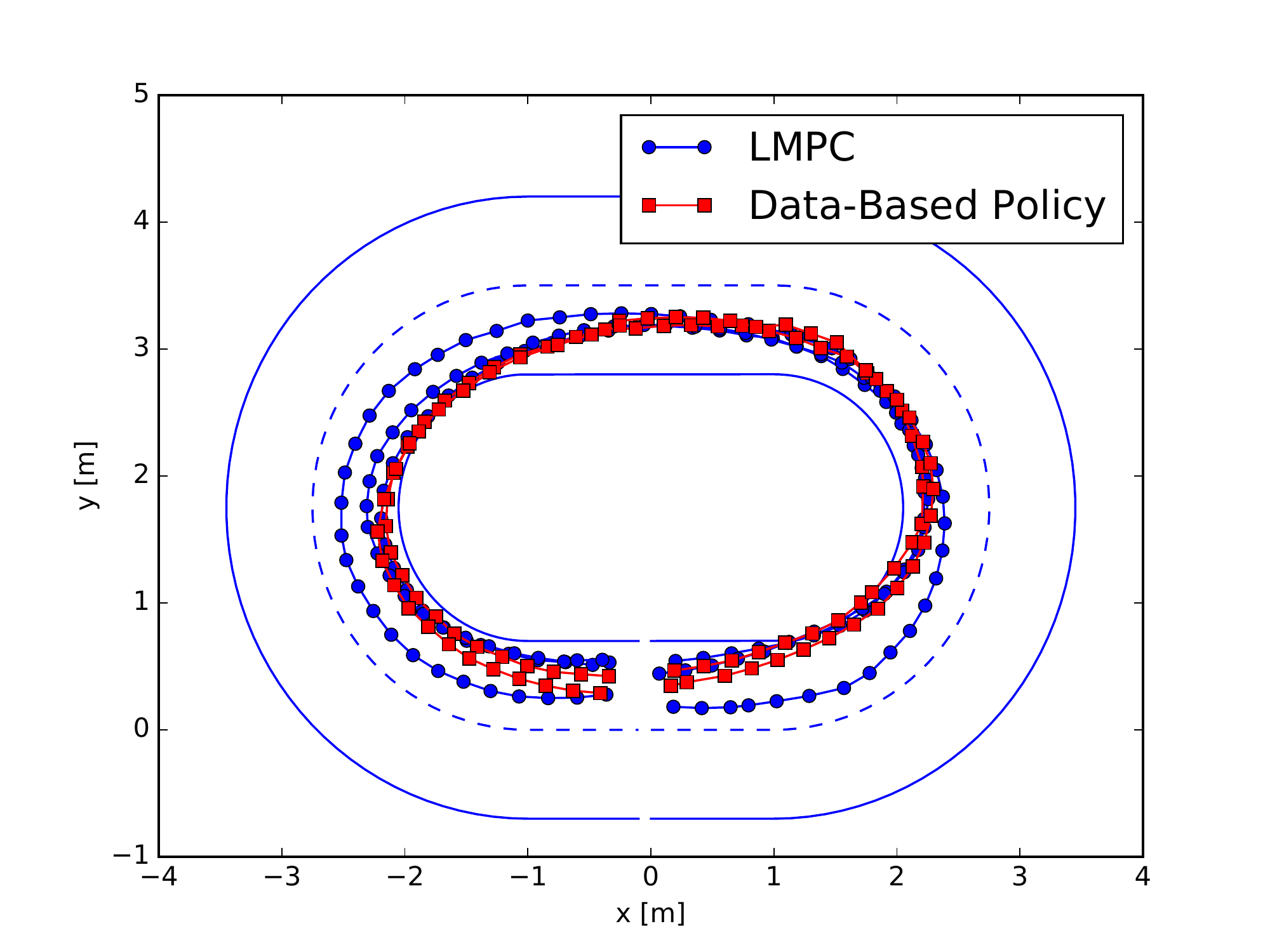}.
    \caption{In red squares are shown the closed-loop trajectories performed by the data-based policy on the oval-shaped track. In blue circles are reported three trajectories in the sampled safe set. Finally, the green dashed line marks the centerline of the track.}
    \label{fig:closedLoopOval}
\end{figure}

For the first $29$ laps of the experiment, we run the Learning Model Predictive Controller (LMPC) from \cite{cdcLMPC} to learn a fast trajectory which drives the vehicle around the track. From the $30$th lap, we run the local data-based policy \eqref{eq:localValueFunc} and \eqref{eq:localPolicy} using the latest $M = 8$ laps and $N = 10$ stored data for each lap. Therefore, the control action is computed upon solving the small optimization problem \eqref{eq:localValueFunc} where $[\lambda_{0|t}^0, \ldots,\lambda_{k|t}^j,\ldots, \lambda_{T_M|t}^M] \in \mathbb{R}^{M|\mathcal{K}^{j}(x)|}$ with $M|\mathcal{K}^{j}(x)| = 80$.

\begin{figure}[h!]
    \centering
	\includegraphics[width= \columnwidth]{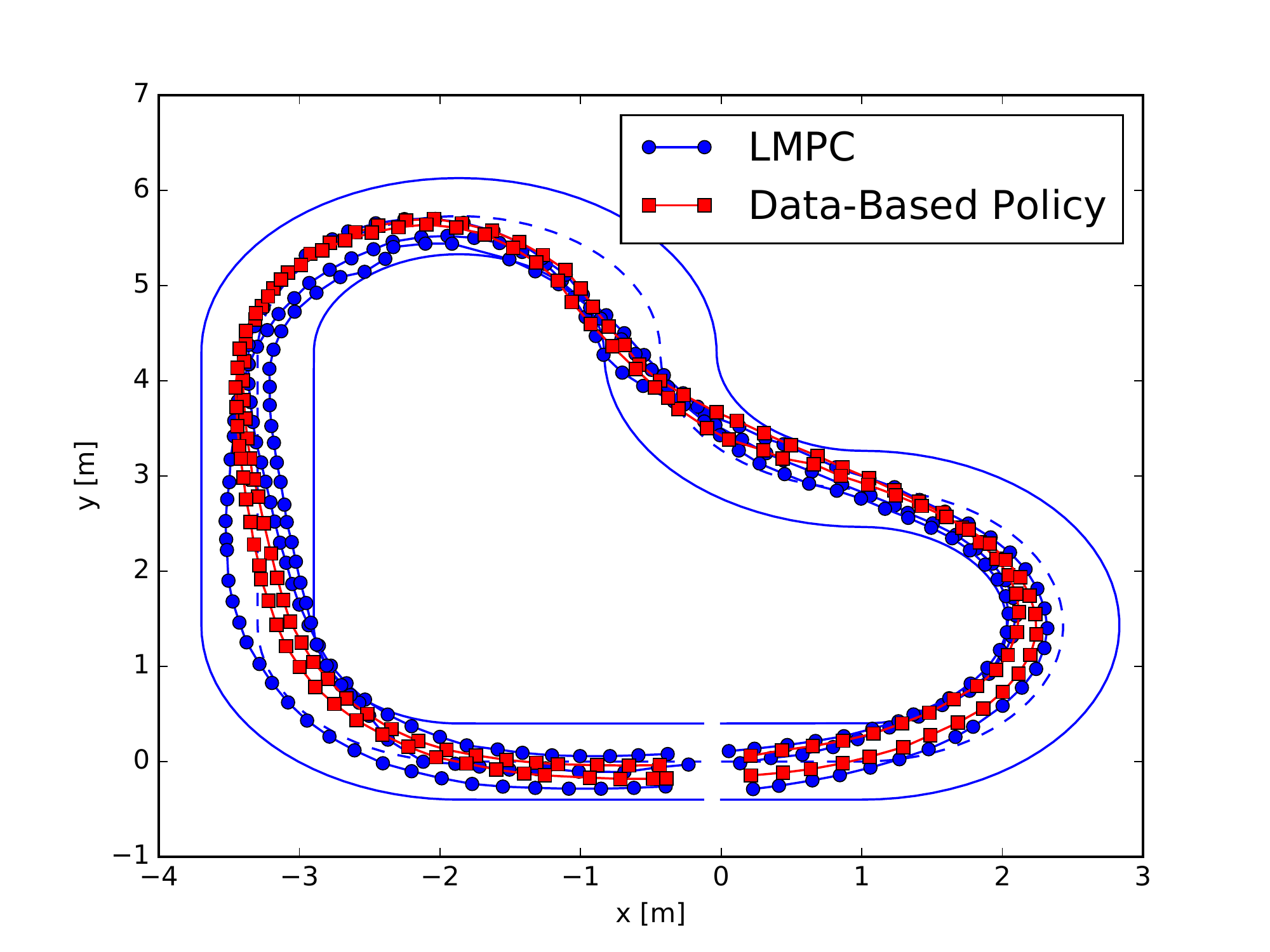}
    \caption{In red squares are shown the closed-loop trajectories performed by the data-based policy on the L-shaped track. In blue circles are reported three trajectories in the sampled safe set. Finally, the green dashed line marks the centerline of the track.}
    \label{fig:closedLoopLshape}
\end{figure}

\begin{figure}[h!]
    \centering
	\includegraphics[width= \columnwidth]{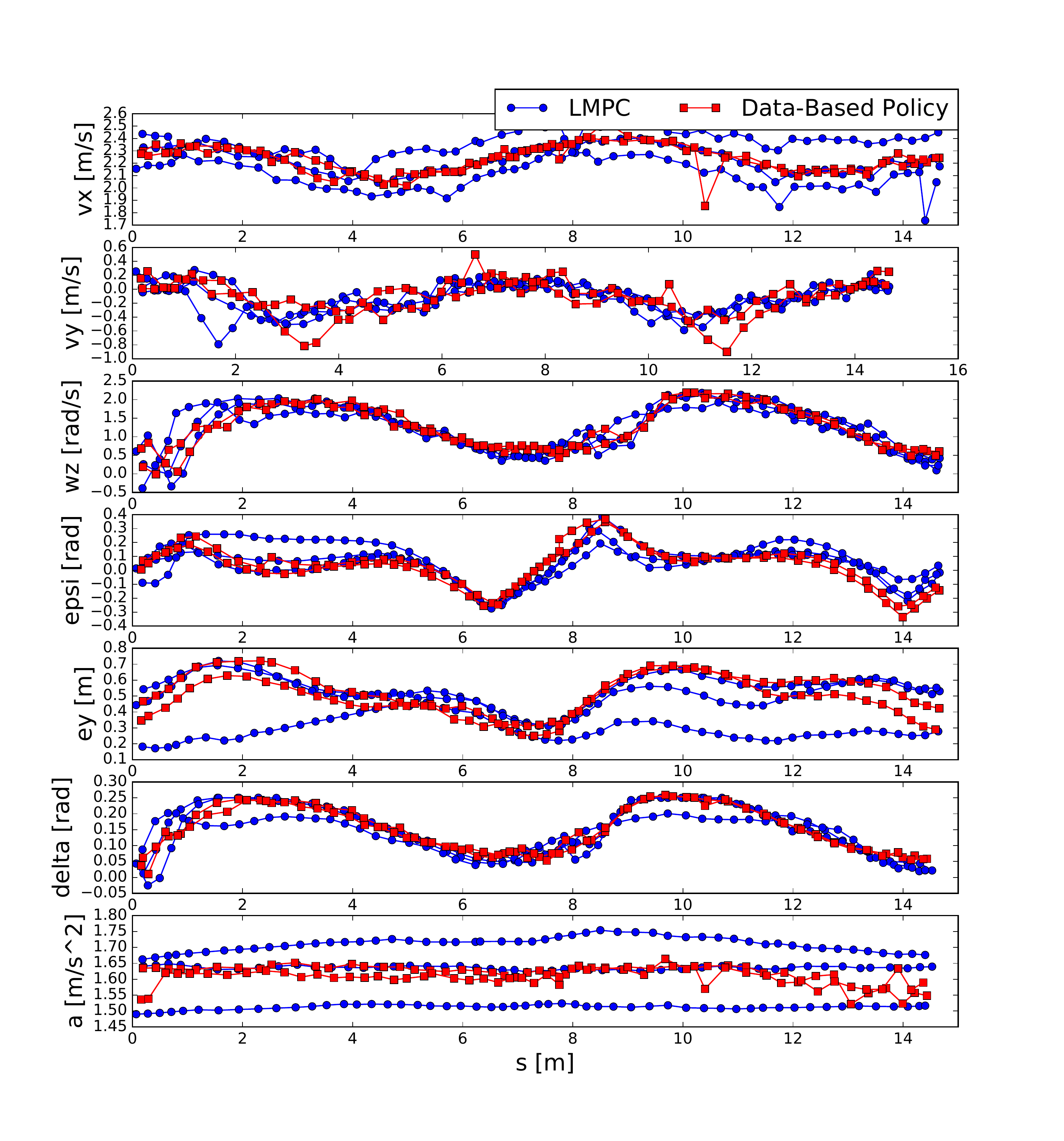}.
    \caption{Closed-loop trajectory and associated inputs of the data-Based policy and LMPC on the oval-shaped track.}
    \label{fig:stateInputOval}
\end{figure}

We tested the controller on an oval-shaped and L-shaped tracks. Figures~\ref{fig:closedLoopOval}-\ref{fig:stateInputLshape} show that the local data-based policy \eqref{eq:localValueFunc} and \eqref{eq:localPolicy} is able to drive the vehicle around the track satisfying input and state constraints. Furthermore, we notice that the closed-loop trajectories generated with the local data-based policy lies in the convex hull of the sampled safe set $\mathcal{SS}$, which is constructed from the last $8$ trajectories performed by the LMPC. It is interesting to notice that the real system is nonlinear but smooth and, for this reason, the system dynamics can be locally linearized. Intuitively, the existence of a local linear model allows us to use the local data-based policy to safely drive the vehicle. Indeed at each time $t$ the controller uses only the historical data close to the system's state $x_t$.

\begin{figure}[h!]
    \centering
	\includegraphics[width= \columnwidth]{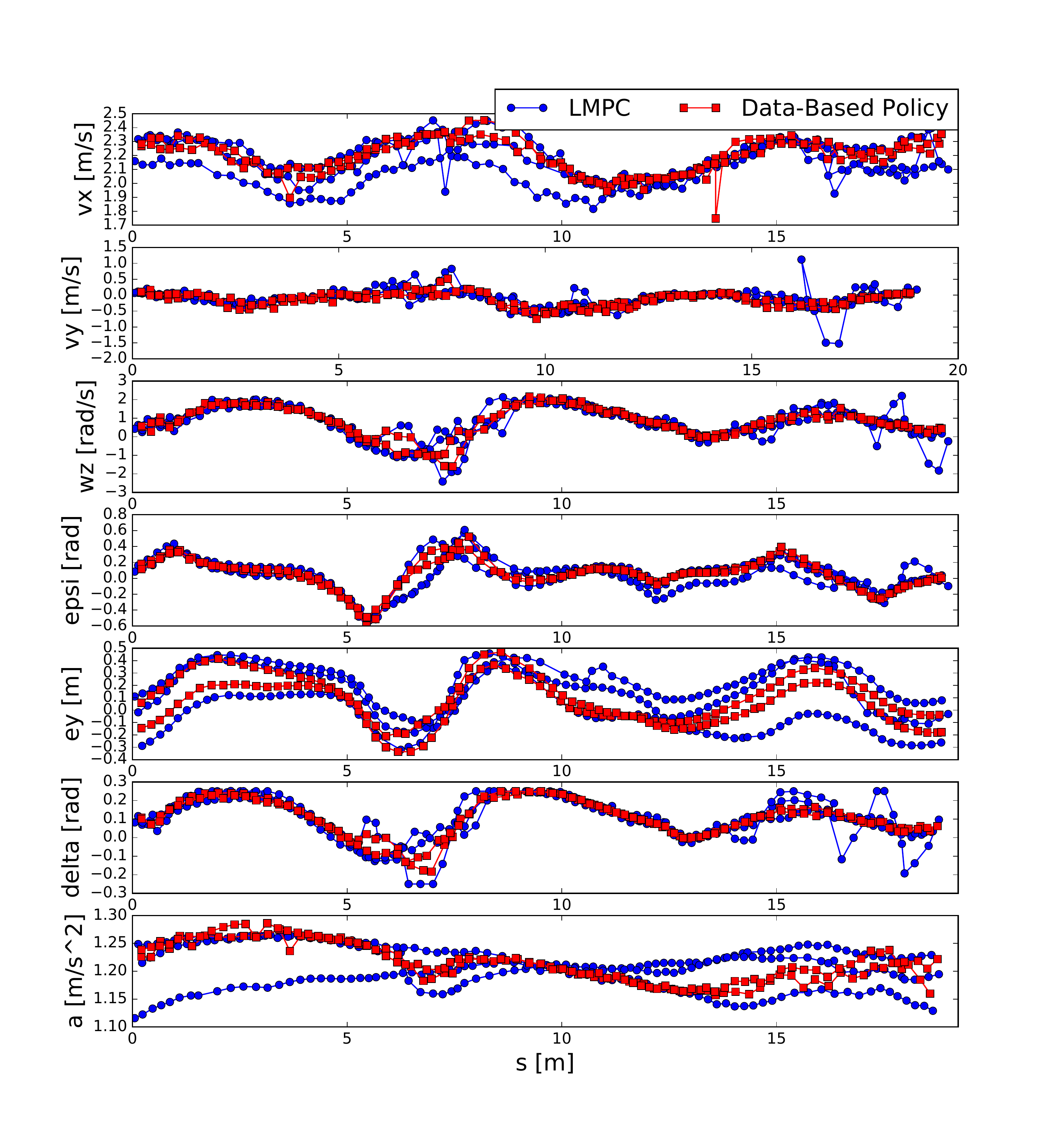}
    \caption{Closed-loop trajectory and associated inputs of the data-Based policy and LMPC on the L-shaped track.}
    \label{fig:stateInputLshape}
\end{figure}

Figures~\ref{fig:lapTimeOval}-\ref{fig:lapTimeLshape} report the lap time over the lap number. We notice that the data-based policy is able to safely drive the vehicle around the track, without hurting the closed-loop performance. In particular, the data-based policy is able to replicate the best lap times performed by the LMPC controller on both tracks. 

Finally, we analyze the computational time. We compare the computational cost associated with the proposed data-based policy and with the LMPC. Table~\ref{table:computationalTime} shows that on average it took $\sim1.3$ms to evaluate the proposed data-based policy and $\sim29.5$ms to evaluate the LMPC policy.

\begin{table}[h!]
\centering\caption{Comparison of computational time}\label{table:computationalTime}
\begin{tabular}{lrrrr}
 \toprule
 $ ~$                           & \text{Avarage}  & \text{Min} & \text{Max}  & Std Deviation\\ \midrule
 $\text{LMPC}$           & $29.5$ms  & $21.8$ms & $50.0$ms & $6.1$ms  \\
 $\text{Data-Based Policy}$     &  $1.3$ms & $1.1$ms & $2.3$ms & $0.2$ms  \\ \bottomrule
\end{tabular}
\end{table}

\begin{figure}[h!]
    \centering
	\includegraphics[width= \columnwidth]{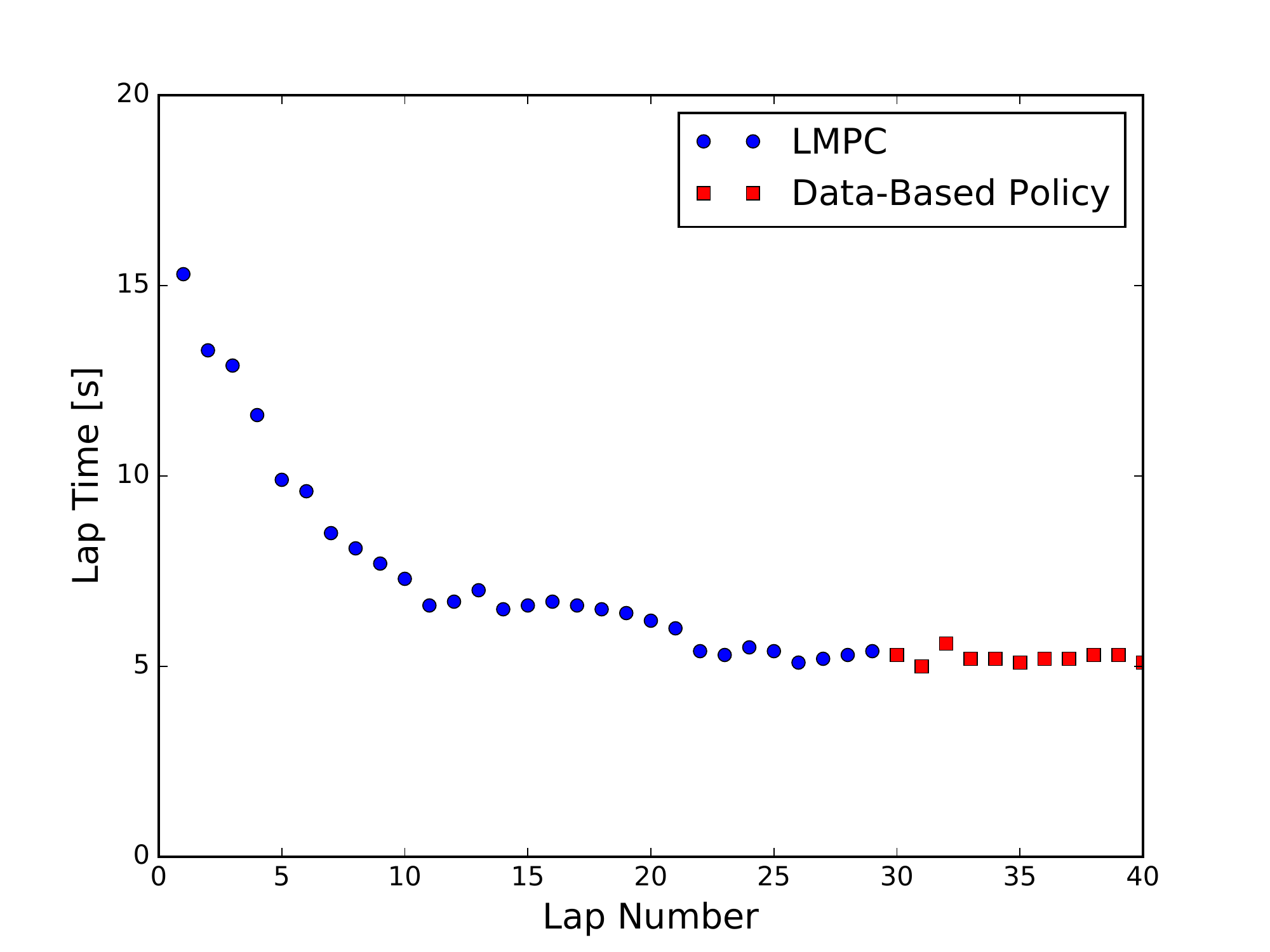}
\caption{Lap time on oval-shaped track over the lap number. At the $30$th lap the data-based policy drives the vehicle around the track without degrading the closed loop-performance.}\label{fig:lapTimeOval}
\end{figure}

\begin{figure}[h!]
    \centering
	\includegraphics[width= \columnwidth]{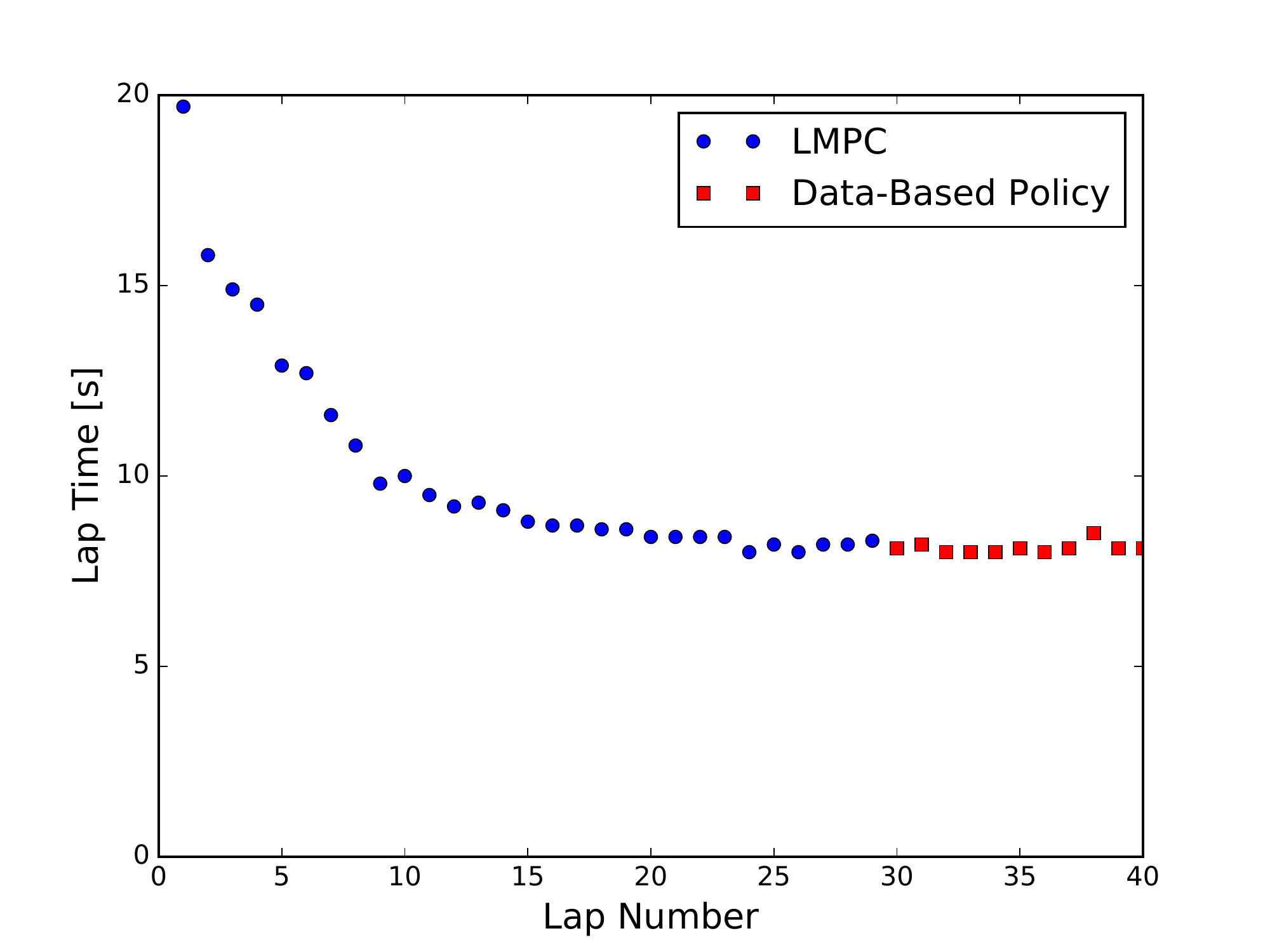}
\caption{Lap time on L-shaped track over the lap number. At the $30$th lap the data-based policy drives the vehicle around the track without degrading the closed loop-performance.}\label{fig:lapTimeLshape}
\end{figure}


\section{Conclusions}\label{sec:conclusions}
In this work we have proposed a simple strategy to construct a data-based policy. Firstly, we used historical data to construct a global and local $Q$-function, which approximates the value function. Afterwards,  we presented the data-based policy evaluates the $Q$-function and computes the control action from the stored input sequences.  We showed that the proposed strategies guarantees safety, stability and performance bounds. Finally, we tested the proposed data-based policy on an autonomous racing example. We show that the proposed strategy matches the performance of our ILC controller, while being $30$x faster at computing the control input. 

\section{Acknowledgment}
Some of the research described in this review was funded by the Hyundai Center of Excellence
at the University of California, Berkeley. This work was also sponsored by the Office of Naval
Research. The views and conclusions contained herein are those of the authors and should not be
interpreted as necessarily representing the official policies or endorsements, either expressed or
implied, of the Office of Naval Research or the US government.

\section{Appendix}
In this Appendix, we show that the proposed data-based policy may be used to steer a linear time invariant system to a terminal invariant set $\mathcal{X}_F$. In order to prove that the properties from Propositions~1-3 hold also in this settings the following assumptions must hold.

\begin{assumption}\label{assAPP:vertices}
The terminal set $\mathcal{X}_F$ is defined by the convex hull of the terminal state of the stored trajectories~\eqref{eq:givenClosedLoop}, i.e.
$\mathcal{X}_F = \text{Conv}\big( \cup_{j=0}^M x_{T_j}^j \big)$.
\end{assumption}

\begin{assumption}\label{assAPP:feasibility}
All $M+1$ input and state sequences in \eqref{eq:givenIinputs}-\eqref{eq:givenClosedLoop} are feasible and known. Furthermore, assume that the state sequence in~\eqref{eq:givenClosedLoop} converges to the terminal set~$\mathcal{X}_F$ and the terminal input $u_{T_j}^j$ keeps the evolution of the system~\eqref{eq:System} into $\mathcal{X}_F$. More formally, we assume that $x_{T_j}^j \in \mathcal{X}_F, \forall j \in \{0,\ldots,M\}$ and $ A x_{T^j} + B u_{T^j} \in \mathcal{X}_F$.
\end{assumption}

\begin{prop}
\textit{(Feasibility)} Consider the closed-loop system~\eqref{eq:System} and \eqref{eq:policy}. Let Assumptions~\ref{assAPP:vertices}-\ref{assAPP:feasibility} hold and $\mathcal{CS}$ be the convex safe set defined in \eqref{eq:CS}. If the initial state $x_0 \in \mathcal{CS}$. Then, the data-based policy \eqref{eq:valueFuncEval} and \eqref{eq:policy} is feasible for all time $t\geq0$.
\end{prop}
\begin{proof}
The proof follows from linearity of the system. \\
We assume that at time $t$ the system state $x_t \in \mathcal{CS}$, therefore the optimization problem \eqref{eq:valueFuncEval} is feasible. Let \eqref{eq:optimalSol} be the optimal solution to \eqref{eq:valueFuncEval}, then at the next time step $t+1$ we have 
\begin{equation*}
\begin{aligned}
         x_{t+1} &=  A x_t + B \sum_{j=0}^M\sum_{k=0}^{T_j} \lambda^{j,*}_{k|t} u^j_k  \\
    &  = A \sum_{j=0}^M\sum_{k=0}^{T_j} \lambda^{j,*}_{k|t} x^j_k + B \sum_{j=0}^M\sum_{k=0}^{T_j} \lambda^{j,*}_{k|t} u^j_k \\
    &  = \sum_{j=0}^M\sum_{k=0}^{T_j} \lambda^{j,*}_{k|t} (A  x^j_k + B u_k^j) \in \mathcal{CS}.
\end{aligned}
\end{equation*}
By Assumption~\ref{assAPP:feasibility} we have that for all $\forall j\in \{0, \ldots M\}$  it exist $\lambda^j_k \geq 0$ such that $\sum_{k=0}^M \lambda_k^j=1$ and
\begin{equation*}
\begin{aligned}
     \sum_{j=0}^M \lambda^{j,*}_{T_j|t} (A  x^j_{T_j} + B u_{T_j}^j)&= \sum_{j=0}^M \lambda^{j,*}_{T_j|t} \sum_{k=0}^M \lambda_k^j x_{T_k}^k \\
     &= \sum_{k=0}^M \sum_{j=0}^M \lambda^{j,*}_{T_j|t} \lambda_k^j x_{T_k}^k =  \sum_{k=0}^M \tilde \lambda_k x_{T_k}^k
\end{aligned}
\end{equation*} 
where $\forall k\in \{0, \ldots M\}$ we defined $\tilde \lambda_k = \sum_{i=0}^M \lambda^{i,*}_{T_i|t} \lambda_k^i$. It follows that
\begin{equation*}
    x_{t+1} =  \sum_{j=0}^M\sum_{k=0}^{T_j} \lambda^{j,*}_{k|t} (A  x^j_k + B u_k^j) = \sum_{j=0}^M\sum_{k=0}^{T_j} \bar \lambda^{j}_k x^j_k
\end{equation*}
where $\forall j\in \{0, \ldots M\}$ 
\begin{equation}
    \begin{aligned}\label{eq:APPfeasibleSol}
        &\bar \lambda_0^j = 0, \\
        &\bar \lambda_{k_j}^j = \lambda_{k_j-1|t}^{j,*}, \quad \quad \quad \quad \quad \forall k_j \in \{ 1, \ldots, T_j-1 \} \\
        &\bar \lambda_{T_j}^j = \lambda_{T_j-1|t}^{j,*} + \tilde \lambda_j
    \end{aligned}
\end{equation}
is a feasible solution to the optimization problem \eqref{eq:valueFuncEval} at time $t+1$.\\
By assumption we have at time $t=0$ the state $x_0 \in \mathcal{CS}$. Furthermore, we have shown that if at time $t$ the state $x_t \in \mathcal{CS}$, then at time $t+1$ the state $x_{t+1} \in \mathcal{CS}$ and the optimization problem \eqref{eq:valueFuncEval} is feasible. Therefore by induction we conclude that $x_t \in \mathcal{CS} \subseteq \mathcal{X}, ~\forall t \in \mathbb{Z}_{0+}$ and that the optimization problem \eqref{eq:valueFuncEval} is feasible $\forall t \in \mathbb{Z}_{0+}$.
\end{proof}

In order to prove convergence we make the following assumption on the stage cost.

\begin{assumption}\label{assAPP:cost}
The stage cost $h(\cdot, \cdot)$ is a continuous convex function and $\forall u \in \mathcal{U}$ it satisfies
\begin{equation}
\begin{aligned}
h(x,u) = 0, \forall x \in \mathcal{X}_F \textrm{ and}~ h(x,u) \succ 0 ~ \forall ~ x \in&~\rr^n \setminus \{\mathcal{X}_F\}. \notag
\end{aligned}
\end{equation}
\end{assumption}

\begin{prop}
\textit{(Convergence)} Consider the closed-loop system~\eqref{eq:System} and \eqref{eq:policy}. Let Assumption~\ref{assAPP:vertices}-\ref{assAPP:cost} hold and $\mathcal{CS}$ be the convex safe set defined in \eqref{eq:CS}. If the initial state $x_0 \in \mathcal{CS}$. Then, the origin of the closed-loop system \eqref{eq:System} and \eqref{eq:policy} is asymptotically stable.
\end{prop}
\begin{proof}
The proof follows from the proof of Proposition~2. In particular, the candidate solution~\eqref{eq:APPfeasibleSol} may be exploited to show that $Q(\cdot)$ is Lyapunov function along the closed-loop trajectory. 
\end{proof}

\begin{prop}
\textit{(Cost)} Consider the closed-loop system~\eqref{eq:System} and \eqref{eq:policy}. Let Assumptions~\ref{assAPP:vertices}-\ref{assAPP:cost} hold and $\mathcal{CS}$ be the convex safe set defined in \eqref{eq:CS}. If the initial state $x_0 \in \mathcal{CS}$. Then, the $Q$-function at $x_0$, $Q(x_0)$, upper bounds the cost associated with the trajectory of closed-loop system~\eqref{eq:System} and \eqref{eq:policy},
\begin{equation*}
    J\big(x_0\big) =   \sum_{k=0}^{\infty} h(x_k, u_k) \leq Q\big(x_0\big)
\end{equation*}
where $\{x_0 , \ldots, x_{t}, \ldots\}$ and $ \{u_0 , \ldots, u_{t}, \ldots\}$ are the closed-loop trajectory and associated input sequence, respectively. 

\end{prop}
\begin{proof}
    The proof follows as in Proposition~3.
\end{proof}

\bibliographystyle{IEEEtran}
\bibliography{IEEEabrv,mybibfile}

\end{document}